\documentclass[a4paper,11pt]{article}
\usepackage{a4wide}
\usepackage[UKenglish]{babel}
\usepackage[UKenglish]{isodate}
\usepackage{amsmath,amssymb,amsthm,mathtools,bbm}
\usepackage{bm}
\usepackage[svgnames,table]{xcolor}
\definecolor{myBlue}{RGB}{27, 82, 140}
\usepackage{microtype}
\usepackage{thm-restate}
\usepackage{enumitem}
\usepackage{complexity}
\usepackage{hyperref}
\usepackage[nameinlink]{cleveref}
\usepackage{nameref}
\hypersetup{colorlinks,allcolors=myBlue}

\renewcommand{\vec}[1]{\ensuremath{\mathbf{#1}}}
\renewcommand\emptyset\varnothing

\newcommand{\2}{\vec{2}}
\newcommand{\1}{\vec{1}}
\newcommand{\0}{\vec{0}}
\renewcommand{\A}{\ensuremath{\mathbf{A}}}
\newcommand{\B}{\ensuremath{\mathbf{B}}}
\newcommand{\X}{\ensuremath{\mathbf{X}}}

\newcommand{\KK}{\ensuremath{K}}
\DeclareMathOperator{\sel}{sel}
\DeclareMathOperator{\ar}{ar}
\DeclareMathOperator{\KG}{KG}
\DeclareMathOperator{\PCSP}{PCSP}
\DeclareMathOperator{\Pol}{Pol}
\DeclareMathOperator{\CSP}{CSP}
\DeclareMathOperator{\NAE}{\mathbf{NAE}}
\DeclareMathOperator{\BNAE}{\mathbf{BNAE}}
\DeclareMathOperator{\CF}{\mathbf{CF}}
\DeclareMathOperator{\LO}{\mathbf{LO}}

\theoremstyle{plain}
\newtheorem{theorem}{Theorem}
\newtheorem{lemma}[theorem]{Lemma}

\newtheorem{corollary}[theorem]{Corollary}
\theoremstyle{definition}
\newtheorem{definition}[theorem]{Definition}

\newlist{lemenumerate}{enumerate}{1}
\setlist[lemenumerate]{label=(\roman*),ref=\thelemma\,(\roman*)}
\crefname{lemenumeratei}{lemma}{lemmas}

\usepackage{todonotes}

\begin{document}

\title{Complexity of approximate conflict-free,  linearly-ordered,\\and nonmonochromatic hypergraph colourings\thanks{An extended abstract of this work appeared in the Proceedings of ICALP 2025~\cite{nvwz25:icalp}. This work was supported by UKRI EP/X024431/1 and a Clarendon Fund Scholarship. Work done while Tamio-Vesa Nakajima was at the University of Oxford. For the purpose of Open Access, the authors have applied a CC BY public copyright licence to any Author Accepted Manuscript version arising from this submission. All data is provided in full in the results section of this paper.}}

\author{Tamio-Vesa Nakajima\\
University of Marburg\\
\texttt{nakajima@uni-marburg.de}
\and
Zephyr Verwimp\\
University of Oxford\\
\texttt{zephyr.verwimp@cs.ox.ac.uk}
\and
Marcin Wrochna\\
University of Warsaw\\
\texttt{m.wrochna@mimuw.edu.pl}
\and
Stanislav \v{Z}ivn\'y\\
University of Oxford\\
\texttt{standa.zivny@cs.ox.ac.uk}
}

\date{\today}
\maketitle
\begin{abstract}
Using the algebraic approach to promise constraint satisfaction problems, we establish complexity classifications of three natural variants of hypergraph colourings: standard nonmonochromatic colourings, conflict-free colourings, and linearly-ordered colourings.

Firstly, we show that finding an $\ell$-colouring of a $k$-colourable $r$-uniform hypergraph is \NP-hard for all constant $2\leq k\leq \ell$ and $r\geq 3$. This provides a shorter proof of a celebrated result by Dinur, Regev, and Smyth~[FOCS'02/Combinatorica'05]. 

Secondly, we show that finding an $\ell$-conflict-free colouring of an $r$-uniform hypergraph that admits a $k$-conflict-free colouring is \NP-hard for all constant $2\leq k\leq\ell$ and $r\geq 4$, except for $r=4$ and $k=2$ (and any $\ell$); this case is solvable in polynomial time. The case of $r=3$ is the standard nonmonochromatic colouring, and the case of $r=2$ is the notoriously difficult open problem of approximate graph colouring.

Thirdly, we show that finding an $\ell$-linearly-ordered colouring of an $r$-uniform hypergraph that admits a $k$-linearly-ordered colouring is \NP-hard for all constant $3\leq k\leq\ell$ and $r\geq 4$, thus improving on the results of Nakajima and \v{Z}ivn\'y~[ICALP'22/ACM ToCT'23].
\end{abstract}

\section{Introduction}\label{sec:intro}

\paragraph{Graph colouring}
Graph colouring is one of the most studied computational problems: Given a graph $G$ and an integer $k$, is there a $k$-colouring, i.e., an assignment of one of $k$ colours to the vertices of the graph so that adjacent vertices are assigned different colours? 

Deciding the existence of a 3-colouring is one of Karp's 21 \NP-complete problems~\cite{Karp1972}. Since finding a graph colouring with the smallest number of colours is \NP-hard, there has been much interest in the \emph{approximate graph colouring} (AGC) problem: Given a graph $G$ that admits a colouring with $k$ colours, find a colouring with $\ell$ colours for some $k\leq\ell$. It is believed that for every constant $3 \leq k \leq \ell$, this problem remains \NP-hard~\cite{GJ76}. While some conditional results are known (e.g.~AGC is \NP-hard if we assume the \emph{$2$-to-1 conjecture with perfect completeness}~\cite{Khot02stoc,Dinur09:sicomp}, or only the $d$-to-1 conjecture~\cite{GS20:icalp}), proving unconditional results seems elusive. The strongest results known so far are for $\ell = 2k - 1$~\cite{BBKO21} and $\ell = \binom{k}{\lfloor k / 2 \rfloor} - 1$~\cite{KOWZ23} (the first result is stronger for $k = 3, 4$, and equal to the second for $k = 5$). As progress on proving the hardness of AGC seems to have hit a barrier, it is natural to try to attack variants of AGC, to see if any of the ideas and insights from those problems could apply to the AGC. In this paper,  we will focus on hypergraph generalisations of graph colourings.

\medskip
An (undirected) \emph{$r$-uniform hypergraph} is a pair $(V, E)$, where $V$ is the vertex set and $E \subseteq V^r$ is a set of $r$-tuples that is closed under permutation of coordinates.  Note that a 2-uniform hypergraph is a graph. A  colouring of a graph $(V,E)$ is an assignment of colours $c(v)$ to the vertices $v \in V$ such that for every edge $(u,v) \in E$ we have $c(u) \neq c(v)$. 
The following three equivalent formulations describe the same condition: For every edge $(u,v)\in E$ we have that (i) the set $ \{ c(u), c(v) \}$ contains at least 2 elements, or (ii) some colour in the multiset $\{\{ c(u), c(v) \}\}$ appears exactly once, or (iii) the largest colour in the multiset $\{\{ c(u), c(v) \}\}$ appears exactly once.\footnote{This notion assumes that the colours are taken from a totally ordered set, e.g.,~the natural numbers.}
When these three definitions are applied to hypergraphs, we get three different notions, namely nonmonochromatic (NAE) colourings, conflict-free (CF) colourings, and linearly-ordered (LO) colourings.
Note that any LO colouring is a CF colouring, and any CF colouring is an NAE colouring. 
The three notions of colourings are different already for 4-uniform\footnote{For 3-uniform hypergraphs, NAE and CF colourings coincide, but LO colourings are different.} hypergraphs.\footnote{Indeed, the edge $(a, b, c, d)$ could be assigned colours $\{\{1, 1, 2, 2\}\}$ in an NAE colouring, but not in the other two; whereas the edge could be assigned colours $\{\{1, 2, 2, 2\}\}$ in a CF colouring, but not in an LO colouring.}

\paragraph{Promise CSPs}
We will now review the most relevant literature on the three variants of hypergraph colourings.
It will be convenient to present the existing results in the framework of so-called \emph{promise constraint satisfaction problems} (PCSPs)~\cite{AGH17,BG21:sicomp}, as we shall use the tools developed for understanding the computational complexity of PCSPs~\cite{BBKO21}.
Constraint satisfaction problems (CSPs) are problems that can be cast as homomorphisms between relational structures. We will only need a special case of relational structures that contain only one relation. Formally, a \emph{relational structure} $\A=(A,{R^{\A}})$ is a pair, where $A$ is the universe of $\A$ and $R^\A\subseteq A^r$ is an $r$-ary relation. By abuse of language, we call $r$ the \emph{arity} of $\A$, and we will often say $\vec{x}$ is a tuple in $\A$ when we mean $\vec{x} \in R^\A$.

Note that graphs and uniform hypergraphs are relational structures, where the universe is the vertex set and the relation is the edge set --- the arity of the relation is 2 for graphs and $r$ for $r$-uniform hypergraphs. A \emph{homomorphism} from one relational structure of arity $r$, say $\A=(A,R^\A)$, to another of the same arity, say $\B=(B,R^\B)$, is a map $h:A\to B$ that preserves the relation: if $(a_1,\ldots,a_r)\in R^{\A}$ then $(h(a_1),\ldots,h(a_r))\in R^{\B}$. We denote the existence of a homomorphism from $\A$ to $\B$ by writing $\A\to\B$.

Given two relational structures $\A$ and $\B$ with $\A\to\B$, the promise constraint satisfaction problem with template $(\A,\B)$, denoted by $\PCSP(\A,\B)$, is the following computational problem. Given a relational structure $\X$ with the promise $\X \to \A$, find a homomorphism from $\X$ to $\B$. This is the search version of the problem. In the decision version, one is given a relational structure $\X$ with the same arity as $\A$ and the task is to output \textsc{Yes} if $\X\to\A$ and \textsc{No} if $X\not\to\B$. Since the decision version reduces to the search version, solving the decision version is no harder than solving the search version. All of our results will hold for both versions --- hardness results will hold even for the decision version, and tractability results will hold even for the search version.

In order to cast approximate hypergraph NAE/CF/LO-colourings as PCSPs, we will need to encode the NAE/CF/LO-colourability of a hypergraph by a homomorphism to a suitable relational structure.
We will thus describe three families of relational structures capturing the three types of hypergraph colourings mentioned above (and therefore implicitly graph colouring).
For any arity $r\geq 2$ and domain size $k$, we define:\footnote{For any positive integer $n$, we write $[n]$ for the set $\{1,2,\ldots,n\}$.}
\begin{align*}
    \NAE^r_k &= (\{0,1,\dots,k-1\}, \{ (x_1, \ldots, x_r) \mid \exists\, i, j \in [r] : x_i \neq x_j \}), \\
    \CF^r_k &= (\{0,1,\dots,k-1\}, \{ (x_1, \ldots, x_r) \mid \exists\, i \in [r] \ \forall\, j\neq i \in [r]:  x_i \neq x_j \}), \\
    \LO^r_k &= (\{0,1,\dots,k-1\}, \{ (x_1, \ldots, x_r) \mid \exists\, i \in [r] \ \forall\, j\neq i \in [r] : x_i > x_j \}).
\end{align*}
Observe that an $r$-uniform hypergraph $\X$ has an NAE $k$-colouring if and only if $\X \to \NAE^r_k$. The analogous statement holds for CF and LO colourings. Since NAE, LO and CF colourings are all identical to graph colouring on uniformity 2, we see that $k$ vs.~$\ell$ AGC is the same as $\PCSP(\NAE^2_k, \NAE^2_\ell)$ --- or equivalently $\PCSP(\CF^2_k, \CF^2_\ell)$ or $\PCSP(\LO^2_k, \LO^2_\ell)$.

Other studied notions of hypergraph colourings include strong hypergraph colouring~\cite{BG16-graph} and rainbow hypergraph colouring~\cite{GL18,ABP20,GS20:rainbow}. 

\paragraph{Nonmonochromatic colourings}
The most studied hypergraph colourings are nonmonochromatic colourings, also known as weak hypergraph colourings.
This is the weakest non-trivial restriction one can impose when colouring the vertices of a hypergraph, i.e., any type of hypergraph colouring (that excludes constant colourings) is also a nonmonochromatic colouring. As mentioned before, nonmonochromatic $k$-colourings of an $r$-uniform hypergraph correspond to homomorphisms from the hypergraph to $\NAE^r_k$. Since nonmonochromatic colouring is \NP-hard for any uniformity \(r \geq 3\) and number of colours \(k \geq 2\), an investigation of the approximate version led to the following result (the arity \(\geq 4\) case was shown in \cite{GHS02-HardnessApproximateHypergraphColoring} without the use of topology, while the arity \(3\) case was shown in \cite{DRS05} relying on the higher chromatic number of Kneser and Schrijver graphs): \begin{restatable}{theorem}{main}\label{thm:main} $\PCSP(\NAE_k^r, \NAE_\ell^r)$ is \NP-hard for all constant $2 \leq k \leq \ell$ and $r \geq 3$.
\end{restatable}

In this paper we will provide a simpler proof of this result. The proof in~\cite{DRS05} relies on constructing a somewhat ad-hoc reduction and analysing its completeness and soundness. We recast this proof in the  recent algebraic framework for PCSPs~\cite{BBKO21}. We also avoid the use of Schrijver graphs, working only with the simpler Kneser graphs, plus a (correct and very easy) case of Hedetniemi’s conjecture\footnote{Note that Hedetniemi’s conjecture was proved false in the general case~\cite{shitov_counterexamples_2019}.} (cf.~\Cref{prodchrombound}). We believe that our simplification is of interest since it replaces a more quantitative analysis of the polymorphisms (i.e. bounds are functions of the arity) with one that only deals with constant bounds everywhere.

\paragraph{Conflict-free colourings}

A conflict-free hypergraph colouring is a colouring of the vertices in a
hypergraph such that every hyperedge has at least one uniquely coloured
vertex~\cite{Even03:sicomp,Smorodinsky2013}. 
As mentioned before, conflict-free $k$-colourings of an $r$-uniform hypergraph correspond to homomorphisms from the hypergraph to $\CF^r_k$.
We shall determine the complexity of $\PCSP(\CF^r_k,\CF^r_\ell)$ for all
constants $2\leq k\leq\ell$ and $r\geq 3$ (the case of $r=3$ corresponding to
nonmonochromatic colourings, i.e., $\CF^3_k=\NAE^3_k$ for every $k$).

After the easy observation that $\PCSP(\CF_k^r, \CF_\ell^r)$ reduces to $\PCSP(\CF_k^{r + t}, \CF_\ell^{r + t})$ for $t \geq 2$ (cf.~\Cref{lem:red} in~\Cref{sec:prelims}), \Cref{thm:main} directly implies \NP-hardness for promise conflict-free colouring for uniformity $r\geq 5$.
The crux of the result is to deal with the case of uniformity $r=4$.
Note that finding a conflict-free colouring of a 4-uniform hypergraph using 2
colours is identical to solving systems of equations of the form $x + y + z + t \equiv 1 \pmod{2}$ over $\mathbb{Z}_2$, and is hence in \P{} and consequently so is $\PCSP(\CF_2^4, \CF_\ell^4)$ for every $\ell\geq 2$. 
We resolve the only remaining case, showing in~\Cref{sec:proof} that $\PCSP(\CF_k^4, \CF_\ell^4)$ is \NP-hard for all $3 \leq k \leq \ell$.
Summarising, we have
\begin{restatable}{theorem}{cf}\label{thm:cf}
$\PCSP(\CF_k^r, \CF_\ell^r)$ is \NP-hard  for all constant $2 \leq k \leq \ell$ and $r \geq 3$, except for $k=2$ and $r=4$, which is in \P.
\end{restatable}

This also immediately implies the following (much weaker) corollary, which does not appear to have been known in the CF-colouring literature.

\begin{corollary}
    It is \NP-hard to approximate the conflict-free chromatic
    number\footnote{That is, the minimum number of colours needed to CF-colour
    a given hypergraph.} of a hypergraph to within any constant factor, even if it
    is $r$-uniform for some constant $r \geq 3$.
\end{corollary}

\paragraph{Linearly-ordered colourings}

A linearly-ordered~\cite{Barto21:stacs} (or
unique-maximum~\cite{Cheilaris13:sidma}) hypergraph colouring is a colouring of
the vertices in a hypergraph with linearly-ordered colours such that the maximum
colour in every hyperedge is unique.
As mentioned before, linearly-ordered $k$-colourings of an $r$-uniform hypergraph correspond to homomorphisms from the hypergraph to $\LO^r_k$.

Barto, Battistelli, and Berg~\cite{Barto21:stacs} conjectured that
$\PCSP(\LO_k^3,\LO_\ell^3)$ is \NP-hard for all constant $2\leq k\leq \ell$.
Building on the topological methods of Krokhin, Opr\v{s}al, Wrochna, and
\v{Z}ivn\'y~\cite{KOWZ23}, Filakovsk\'y, Nakajima, Opr\v{s}al, Tasinato, and
Wagner~\cite{fnotw25:toct} established \NP-hardness of $\PCSP(\LO_3^3,\LO_4^3)$.
This result was recently subsumed by \NP-hardness of $\PCSP(\LO_2^3,\LO_3^3)$
shown by Krokhin and Vagnozzi~\cite{krokhin_approximating_2025} (via a simple
reduction that gives \NP-hardness of $\PCSP(\LO_k^3,\LO_{k+1}^3)$ for all $k\geq
2$).

For higher arities ($r\geq 4$), Nakajima and \v{Z}ivn\'y~\cite{NZ23:toct} showed \NP-hardness of
$\PCSP(\LO_k^r,\LO_\ell^r)$ for every $2\leq k\leq \ell$ whenever $r\geq \ell-k+4$. We strengthen this result, showing \NP-hardness of $\PCSP(\LO_k^r,\LO_\ell^r)$ for $3 \leq k \leq \ell$ for every $r \geq 4$.

\begin{restatable}{theorem}{lo}\label{thm:lo}
    $\PCSP(\LO^r_k,\LO^r_\ell)$ is \NP-hard for all constant $3\leq k\leq \ell$ and $r\geq 4$.
\end{restatable}

Observe that this theorem covers nearly all the cases from the result
of~\cite{NZ23:toct}: the only case not covered is $k = 2$ and $r \geq \ell + 2$.
In particular, \Cref{thm:lo} has no requirement on $r$ in terms of $\ell$,
unlike the result in~\cite{NZ23:toct}. Indeed, \Cref{thm:lo} covers the full range of
parameters except for the cases $r = 3$ or $k = 2$ (and thus the conjecture
of Barto et al. remains open).

It is worth digressing somewhat to discuss the appearance of topological methods
within these proofs. Our proof uses the chromatic number of the Kneser graph as
an essential ingredient --- this is a topological fact, and thus our proof is in
some sense topological. (This is similar to the appearance of topology within
hardness proofs for rainbow colourings~\cite{ABP20}.) On the other hand, the
topological approach of~\cite{fnotw25:toct}, which proved that $\PCSP(\LO_3^3, \LO_4^3)$ is \NP-hard, is rather different. It assigns each relational structure an equivariant simplicial complex in a ``nice enough'' way so that the topological properties of these simplicial complexes imply the hardness of the original template. It would be interesting to see if these two approaches can be merged, or combined to strengthen both.

\section{Preliminaries}
\label{sec:prelims}

Let $(\A,\B)$ be a PCSP template with $\A$ and $\B$ of arity $r$.
A \emph{polymorphism} of arity $n=\ar(f)$ of $(\A,\B)$ is a function $f:A^n \to B$ such
that if  $f$ is applied component-wise to any $n$-tuple of elements of $R^\A$ it
gives an element of $R^\B$. In more detail, whenever $(a_{ij})$ is an $r\times
n$ matrix such that every column is in $R^\A$, then $f$ applied to the rows
gives an $r$-tuple which is in $R^\B$.
We denote by $\Pol^{(n)}(\A,\B)$ the collection of $n$-ary polymorphisms of $(\A, \B)$, and we let $\Pol(\A, \B) = \bigcup_n \Pol^{(n)}(\A, \B)$.

For an $n$-ary function $f:A^n \to B$ and a map $\pi : [n] \to [m]$, we say that an $m$-ary function $g:A^m\to B$ is the \emph{minor of $f$ given by $\pi$} if $g(x_1,\ldots,x_m) = f(x_{\pi(1)},\ldots,x_{\pi(n)})$.
We write $f \xrightarrow{\pi} g$ if $g$ is the minor of $f$ given by $\pi$. Note that $\Pol(\A,\B)$ is closed under minors.

We use $\leq_p$ to denote a polynomial-time many-one reduction.

\begin{theorem}[\cite{BG21:sicomp}]\label{polsubset}
If $\Pol(\A,\B) \subseteq \Pol(\A',\B')$ then $\PCSP(\A',\B')\leq_p\PCSP(\A,\B)$.
\end{theorem}
\begin{lemma}\label{lem:red}
    For any $t\geq 2$, $\PCSP(\CF_k^r, \CF_\ell^r)
    \leq_p \PCSP(\CF_k^{r + t}, \CF_\ell^{r + t})$.
\end{lemma}
\begin{proof}
We will use~\Cref{polsubset} and show that $\Pol(\CF_k^{r + t}, \CF_\ell^{r + t}) \subseteq \Pol(\CF_k^r, \CF_\ell^r)$. Suppose $f \in \Pol^{(n)}(\CF_k^{r + t}, \CF_\ell^{r + t})$. Consider any $r \times n$ matrix $A$ with rows $\vec{a}_1,\ldots,\vec{a}_r \in [k]^n$ such that every column in the matrix has a unique entry. We can choose $\vec{b} \in [k]^n$ such that each column in the $(r + t) \times n$ matrix $A'$ with rows $\vec{a}_1,\ldots,\vec{a}_r$, and $t$ copies of $\vec{b}$ has a unique entry: choose element $i$ of $\vec{b}$ to be any value in $[k]$ \emph{other than} the unique entry in the $i$-th column of $A$.
Since $f$ is a polymorphism of $(\CF_k^r,\CF_\ell^r)$, we get that $(f(\vec{a}_1),\ldots,f(\vec{a}_r),f(\vec{b}),\ldots,f(\vec{b}))$ has a unique entry. Since $t\geq 2$, $(f(\vec{a}_1),\ldots,f(\vec{a}_r))$ must also have a unique entry. Thus $f \in \Pol^{(n)}(\CF_k^{r}, \CF_\ell^{r})$ as required.
\end{proof}

An \emph{$\ell$-chain of minors} is a sequence of the form $f_0 \xrightarrow{\pi_{0,1}} f_1 \xrightarrow{\pi_{1,2}} \cdots \xrightarrow{\pi_{\ell-1,\ell}} f_\ell$. We shall then write $\pi_{i,j}: [\ar(f_i)]\to[\ar(f_j)]$ for the composition of $\pi_{i,i+1}, \dots, \pi_{j-1,j}$, for any $0\leq i < j\leq \ell$. Note that $f_i \xrightarrow{\pi_{i,j}} f_j$. 
We shall use the following \NP-hardness criterion for PCSPs. 
\begin{theorem}[\cite{BWZ21}]\label{thm:chain hardness}
    Suppose there are constants $k, \ell$ and an assignment
    $\sel$ which, for every $f \in \Pol(\A, \B)$, outputs a set $\sel(f) \subseteq [n]$ of size at most $k$, where $n$ is the arity of $f$.
    Suppose furthermore that for every $\ell$-chain of minors there is a pair
    $i<j$ such that $\pi_{i, j}(\sel(f_i))
    \cap \sel(f_j) \neq \varnothing$. Then, 
    $\PCSP(\A,\B)$ is \NP-hard.
\end{theorem}

For a graph $G$, the \emph{chromatic number} of $G$, denoted by $\chi(G)$, is the smallest $k$ such
that $G\to \KK_k$, where $\KK_k$ is the clique on $k$
vertices.
We will rely on Lov\'asz's result for the chromatic number of Kneser graphs~\cite{Lovasz78}.
For $1\leq h \leq |A|$, write $A^{(h)}$ for the family of subsets of $A$ of size $h$.
The \emph{Kneser graph} is defined as 
$\KG(A,h) = (A^{(h)},E)$,
where $\{P,Q\} \in E$ if and only if $P\cap Q = \emptyset$. For the special case $A = [n]$, we use the notation $\KG(n, h) = \KG([n], h)$.

\begin{theorem}[\cite{Lovasz78}]\label{kneserchrombound}
  $\chi(\KG(n,h))=n-2h+2$ for any $n \in \mathbb{N}$, $ 1\leq h \leq n/2$.
\end{theorem}
\begin{lemma}\label{prodchrombound}
	Let $\chi(G) > n$. Then $\chi(G \times \KK_{n+1}) > n$.\footnote{Here $\times$ denotes the tensor product of two graphs.}
\end{lemma}
\begin{proof}
	We show the contrapositive: suppose there is a homomorphism $G \times \KK_{n+1} \to \KK_n$.
	Equivalently, there is a homomorphism $G \to \KK_n^{\KK_{n+1}}$, where $\KK_n^{\KK_{n+1}}$ is the graph with vertex set $\{f:[n+1]\to[n]\}$, and $f$ and $g$ adjacent if for every distinct $i,j \in [n+1]$, $f(i) \neq g(j)$. This is equivalent since every homomorphism 
    $f : G \times \KK_{n + 1} \to \KK_n$
     corresponds uniquely to the homomorphism $f' : G \to \KK_n^{\KK_{n + 1}}$ given by $f'(u) = (v \mapsto f(u, v))$.
    There is also a homomorphism $\KK_n^{\KK_{n+1}} \to \KK_n$: map any $f \in V(\KK_n^{\KK_{n+1}})$ to an arbitrary repeating element in the range of $f$ (at least one must exist by the pigeonhole principle). Thus $G \to \KK_n$.
\end{proof}

If $\A\to\A'\to\B'\to\B$ then $(\A,\B)$ is a \emph{homomorphic relaxation} of $(\A',\B')$. In this case it follows from the definitions that  $\PCSP(\A,\B)\leq_p\PCSP(\A',\B')$~\cite{BBKO21}.

\section{Proofs of hardness}
\subsection{Avoiding sets imply hardness}
\label{sec:avoiding}

Our hardness proofs will revolve around the notion of \emph{avoiding sets} for
polymorphisms~\cite{BBKO21}, defined below. For $X \subseteq [n]$, we denote by $\1_X$ the
\emph{indicator vector} of $X$: it is the $n$-dimensional vector with $(\1_X)_i = 1$ for $i \in X$ and $(\1_X)_i = 0$ for $i \in [n]\setminus X$.

\begin{definition}
Take $A$ so that $\{0, 1 \} \subseteq A$. Let $f:A^n\to B$ and $T\subseteq B$.
  A \emph{$T$-avoiding set} for $f$ is a set $P \subseteq [n]$ such that for any
  $R \supseteq P$, we have
  $f(\1_R) \notin T$.
  For $t\in
\mathbb{N}$, we call a set $P$ $t$-avoiding for $f$ if it is $T$-avoiding for
$f$ for some subset $T\subseteq B$ of size $t$. 
\end{definition}

We will first collect some simple properties of avoiding sets.
\begin{lemma}\label{lem:avoidingProperties}
    Let $f : A^n \to B$ and $\ell = |B|$.
    \begin{lemenumerate}
        \item\label{item:noLAvoiding} There are no $\ell$-avoiding sets for $f$.
        \item\label{item:domainAvoiding} $[n]$ is an $(\ell - 1)$-avoiding set for $f$.
        \item\label{item:upwardsClosed} If $U$ is $T$-avoiding for $f$ then so is every $V\supseteq U$.
        \item\label{item:minorPreserving} Take $\pi : [n] \to [m]$ and suppose $f \xrightarrow{\pi} g$ for some $g : A^m \to B$. Suppose $P \subseteq [n]$ is $T$-avoiding for $f$. Then $\pi(P)$ is $T$-avoiding for $g$.
    \end{lemenumerate}
\end{lemma}
\begin{proof}
  For $(i)$, observe that any $\ell$-avoiding set would imply that $f(\1_{[n]})
    \not \in B$, which is impossible. For $(ii)$, note that $[n]$ is $(B \setminus \{
      f(\1_{[n]}) \})$-avoiding, and hence $(\ell - 1)$-avoiding. $(iii)$ follows
    from the definitions.
    For $(iv)$, first observe that $g(\1_X) = f(\1_{\pi^{-1}(X)})$. For
    contradiction, suppose $\pi(P)$ is not $T$-avoiding for $g$, i.e., there exists $R \supseteq \pi(P)$ with $g(\1_R) \in T$. Then $\pi^{-1}(R) \supseteq P$, and furthermore $f(\1_{\pi^{-1}(R)}) = g(\1_R) \in T$. Hence $P$ is not $T$-avoiding for $f$.
\end{proof}

To apply~\Cref{thm:chain hardness}, we want to build $\sel(f)$ out of (small)
avoiding sets for $f$. This is a good idea because avoiding sets are preserved
by minors, as shown in~\Cref{item:minorPreserving}. The issue is that we might have too many
avoiding sets. For the polymorphisms in this paper, many (small) $t$-avoiding
sets which are pairwise disjoint imply the existence of a (small)
($t+1$)-avoiding set. Thus, since there can be no sets that avoid every output
in the range, as shown in~\Cref{item:noLAvoiding}, there must be some maximal $t$ for which a (small) avoiding set exists. By maximality, there cannot be too many disjoint $t$-avoiding sets. Thus, we can build $\sel(f)$ out of these disjoint ``maximally avoiding'' sets.

\begin{theorem}\label{our hardness condition}

Let $(\A,\B)$ be a PCSP template with $\{0, 1\} \subseteq A$ and $\ell = |B|$.
Suppose that there exist constants $N,\{\alpha_t\}_{t=1}^\ell,\{\beta_t\}_{t=1}^\ell$ such that every $f \in \Pol(\A,\B)$ has the following properties:
\begin{enumerate}
    \item $f$ has a 1-avoiding set of size $\leq \beta_1$.
    \item If $f$ is of arity $\geq N$ and has a disjoint family of $> \alpha_t$ many $t$-avoiding sets, all of size $\leq \beta_t$, then $f$ has a $(t+1)$-avoiding set of size $\leq \beta_{t+1}$.
\end{enumerate}
Then, $\PCSP(\A,\B)$ is \NP-hard.
\end{theorem}
\begin{proof}
For each $f \in \Pol(\A,\B)$, define $t(f)$ to be the maximal $t$ such that $f$
  has a $t$-avoiding set of size $\leq \beta_t$. By Assumption 1 and the lack 
  of $\ell$-avoiding sets (cf.~\Cref{item:noLAvoiding}),  $t(f)$ exists and $1\leq t(f) <
  \ell$. For each $f\in \Pol(\A,\B)$, let $\mathcal{F}_f$ be a maximal disjoint family of $t(f)$-avoiding sets of size $\leq \beta_{t(f)}$. Define $\sel(f) = \bigcup \mathcal{F}_f$. Then by Assumption 2, $|\sel(f)| \leq \max\{N,\max_{1\leq t<\ell} \alpha_t\beta_t\} \eqqcolon k$.
  
Using~\Cref{thm:chain hardness}, it remains to show that for every $\ell$-chain
  of minors there are $i, j$ such that $\pi_{i,
  j}\left(\sel\left(f_i\right)\right) \cap \sel\left(f_j\right) \neq \emptyset$.
  Let $f_0 \xrightarrow{\pi_{0,1}} f_1 \xrightarrow{\pi_{1,2}} \cdots
  \xrightarrow{\pi_{\ell-1, \ell}} f_{\ell}$ be such a chain. Since $1\leq
  t(f_i) < \ell$, there are distinct $i,j$ such that $t(f_i)=t(f_j) \eqqcolon
  t$. By~\Cref{item:minorPreserving}, for every $t$-avoiding set $P$ of
  $f_i$, $\pi_{i,j}(P)$ is $t$-avoiding for $f_j$. Hence by maximality of $\mathcal{F}_{f_j}$, $\pi_{i,j}(P)$ intersects $\sel(f_j)$. Since $\sel(f_i)$ is the union of such sets, certainly $\pi_{i,j}(\sel(f_i))$ intersects $\sel(f_j)$.
\end{proof}
The rest of this paper will show hardness of certain PCSP templates by~\Cref{our hardness condition}.

\subsection{Hardness of promise nonmonochromatic colouring}

In this section we prove the advertised hardness results for NAE colouring. The most important case is $\PCSP(\NAE_2^3, \NAE_\ell^3)$ for $\ell \geq 2$; all other cases derive from this one by either gadget reductions or homomorphic relaxations.

\begin{lemma}\label{lem:lem1}
	Let $\ell \geq 2$ and $n \in \mathbb{N}$.
    Any $f \in \Pol^{(n)}(\NAE_2^3, \NAE_\ell^3)$ has a 1-avoiding set of size $ \leq \ell$.
\end{lemma}
\begin{proof}
  By~\Cref{item:domainAvoiding}, $[n]$ is an $(\ell-1)$-avoiding set and hence a 1-avoiding set. Thus if $n\leq\ell$, we are done.
	Otherwise assume $n > \ell$.
	Let $h = \left\lceil \frac{n-\ell}{2}\right\rceil$.
	Consider the Kneser graph $\KG(n,h)$, and colour each vertex $P$ by $f(\1_P)$.
	By~\Cref{kneserchrombound}, $\chi(\KG(n,h))=n-2h+2 > \ell$, so
	there are disjoint sets $P,Q \in [n]^{(h)}$ with the same colour $f(\1_P)=f(\1_Q) \eqqcolon b$.
	Let $X \coloneq [n] \setminus (P \cup Q)$.
    
    We claim that $X$ is a $\{b\}$-avoiding set and thus 1-avoiding.
    Since $|X| = n - 2h \leq \ell$ this completes the proof.
    In order to prove the claim, consider any $Y$ such that $X \subseteq Y
    \subseteq [n]$; we want to show that $f(\1_Y) \neq b$. Construct a matrix in
    which each column corresponds to an element $i\in [n]$ and 
    whose rows are $\1_P, \1_Q, \1_Y$. 
    Observe that since $P \cup Q \cup Y = [n]$ (as $P, Q, X$ is a
    partition of $[n]$), no column contains only 0s. Similarly since $P \cap Q
    \cap Y = \emptyset$ (as $P$ and $Q$ are disjoint), no column contains only
    1s. Hence all the columns of the matrix whose rows are $\1_P, \1_Q, \1_Y$
    are tuples of $\NAE_2^3$. Since $f$ is a polymorphism of
    $(\NAE_2^3,\NAE_\ell^3)$, 
    $(f(\1_P), f(\1_Q), f(\1_Y))$ must be a tuple in $\NAE_\ell^3$. But $f(\1_P) = f(\1_Q) = b$, so $f(\1_Y) \neq b$ as required.
	  Thus $X$ is 1-avoiding. 
\end{proof}

\begin{lemma}\label{lem:lem2}
	Let $1\leq t < \ell$ and $n \geq (\ell+1) \ell^t + \ell+1$.
  Suppose $f \in \Pol^{(n)}(\NAE_2^3, \NAE_\ell^3)$ has $>\binom{\ell}{t}\cdot \ell$ disjoint $t$-avoiding sets of size $\leq \ell^t$.
	Then $f$ has a $(t+1)$-avoiding set of size $\leq \ell^{t+1}$.	
\end{lemma}
\begin{proof}
	By assumption and the pigeonhole principle, $f$ has $\geq \ell+1$ disjoint sets $S_1,\dots,S_{\ell+1}\subseteq [n]$ of size $\leq \ell^t$ that avoid the same $T \subseteq \{0,1,\dots\ell-1\}$ of size $|T|=t$.
	Let $R \coloneq [n] \setminus (S_1 \cup \dots \cup S_{\ell+1})$.
	Let $h = \left\lceil \frac{|R|-\ell}{2}\right\rceil$. We have $h \geq 0$ by the lower bound on $n$.
    
	Consider subsets of $[n]$ which are the union of exactly one of $S_1,\dots,S_{\ell+1}$, and a subset of $R$ of size $h$; let $\mathcal{S}$ be the collection of such subsets. We want to find two disjoint sets $P, Q \in \mathcal{S}$ such that $f(\1_P) = f(\1_Q)$.
    Observe that there is a bijection between $\mathcal{S}$ and the vertex set of $\KK_{
    \ell + 1} \times \KG(R, h)$, given by taking vertex $(i, A)$ of $\KK_{
    \ell + 1} \times \KG(R, h)$ to $S_i \cup A \in \mathcal{S}$. Furthermore, this bijection extends to an isomorphism between the graph  $\KK_{
    \ell + 1} \times \KG(R, h)$ and the graph whose vertex set is $\mathcal{S}$ and which considers $P, Q \in \mathcal{S}$ to be adjacent if and only if they are disjoint.
	By~\Cref{prodchrombound}, $\chi(\KK_{\ell+1} \times \KG(R,h)) > \ell$, so there must exist disjoint sets $P, Q \in \mathcal{S}$ with $f(\1_P) = f(\1_Q) \eqqcolon b$.
    
	Since $S_i \subseteq P$ for some $i$, we have $b \not\in T$.
	Let $X \coloneqq [n] \setminus (P \cup Q)$. Identically to the reasoning in
  the proof of~\Cref{lem:lem1}, observe that for any
	$Y \subseteq [n]$ with $X\subseteq Y$, $f(\1_{Y}) \neq b$.
	Moreover, $X\subseteq Y$ implies  $S_j \subseteq Y$ for
  $(\ell+1)-2 \geq 1$ different values of $j$ (i.e.~those for which $S_j \not \in \{P, Q\}$), hence $f(\1_{Y}) \not \in T$.
	Thus $X$ is $(T \cup \{b\})$-avoiding, so $(t + 1)$-avoiding.
	Finally $|X| \leq ((\ell+1)-2) \cdot \ell^t + |R|-2h \leq (\ell-1) \cdot \ell^t + \ell \leq \ell^{t+1}$.
\end{proof}

\main*

\begin{proof}
$\PCSP(\NAE_2^3,\NAE_\ell^3)$ is \NP-hard for all $\ell \geq 2$ by~\Cref{our hardness condition}, \Cref{lem:lem1}, and~\Cref{lem:lem2}.
To extend \NP-hardness to larger uniformities $r\geq3$, it is sufficient to observe that the map sending a 3-uniform hypergraph $I = (V,E)$ to the $r$-uniform hypergraph $I'=(V,E')$ where $E' = \{(x,y,z,\ldots,z)\mid (x,y,z) \in E\}$ describes a reduction $\PCSP(\NAE_2^{3},\NAE_\ell^{3}) \leq_p \PCSP(\NAE_2^{r}, \NAE_\ell^{r})$.
Finally, we can extend \NP-hardness to any pair $k\leq\ell$ since by homomorphic relaxation
$\PCSP(\NAE_2^r,\NAE_\ell^r)\leq_p\PCSP(\NAE_k^r,\NAE_\ell^r)$.
\end{proof}

\subsection{Hardness of promise conflict-free and linearly-ordered colouring}
\label{sec:proof}

In this section we prove the advertised hardness results for both LO and CF
colourings. For the CF colourings, by~\Cref{lem:red} it suffices to establish
hardness for $r=4$. However, our proof is the same for any $r\geq 4$ and thus we
will present it that way. Since $\LO^r_k\to\CF^r_k$, we can then do both LO and
CF colourings ``in one go'' by proving the hardness of $\PCSP(\LO_3^r, \CF_\ell^r)$ for all $\ell \geq 3$ and $r \geq 4$. 
In the following proofs, we let $\0$ denote the vector whose elements are all 0, and we let $\2_X$ denote a ``scaled indicator vector'': $\2_X = 2 \cdot \1_X$.

\begin{lemma}\label{lem:existence}
    Let $\ell \geq 3$ and $r\geq 4$.
    Then any $f \in \Pol^{(n)}(\LO_3^r, \CF_\ell^r)$ has a 1-avoiding set of size $\leq \ell$.
\end{lemma}
\begin{proof}
  By~\Cref{item:domainAvoiding}, $[n]$ is an $(\ell-1)$-avoiding set and hence a 1-avoiding set. Thus if $n\leq\ell$, we are done.
  Otherwise assume $n > \ell$, and set $h = \lceil \frac{n-\ell}{2}\rceil$.
    Consider the Kneser graph $\KG(n,h)$, and colour each vertex $P$ by $f(\2_P)$. By~\Cref{kneserchrombound}, $\chi(\KG(n, h)) = n - 2h + 2 > \ell$, so there exist disjoint sets $P, Q \in [n]^{(h)}$  such that $f(\2_P) = f(\2_Q)$. Let $X \coloneqq [n]\setminus (P \cup Q)$.
    
    Now observe that for any $Y \subseteq [n]$ with $X\subseteq Y$, 
     all the columns of the $r\times n$ matrix whose rows are $\1_Y$, $\2_P, \2_Q$, and $r-3$ many copies of $\0$
    are tuples of $\LO_3^r$.
    For $i \in P$, the unique maximum is a 2 in the $\2_P$ row; for
    $i \in Q$ the unique maximum is a 2 in the $\2_Q$ row; and for $i \in X$
    the unique maximum is a 1 in the $\1_Y$ row. Hence, since $f$ is a
    polymorphism of $(\LO_3^r,\CF_\ell^r)$, we have that $(f(\1_Y), f(\2_P),
    f(\2_Q), f(\0), \ldots, f(\0))$ is a tuple in $\CF_\ell^r$. Since $f(\2_P) = f(\2_Q)$ we deduce that $f(\1_Y) \neq f(\0)$.
    Thus it follows that $X$ is $\{ f(\0) \}$-avoiding, and thus 1-avoiding. Noting that $X$ has size $n-2h\leq \ell$ completes the proof.
\end{proof}

\begin{lemma} \label{lem:boosting}
  Let $  \ell \geq 3$ and $r\geq 4$. Let $1\leq t<\ell$, and suppose that $f\in \Pol^{(n)}(\LO_3^r, \CF_\ell^r)$ has a disjoint family of $> \binom{\ell}{t}\ell$ many $t$-avoiding sets, all of size $\leq t\ell$. Then $f$ has a $(t+1)$-avoiding set of size $\leq (t+1)\ell$.
\end{lemma}
\begin{proof}
By assumption and the pigeonhole principle, there is a family $\mathcal{F}$ of at least $\ell+1$ disjoint $T$-avoiding sets of size $\leq t \ell$ for some $T \subseteq \{0,1,\dots,\ell-1\}$ of size $t$. Let $h,P,Q,X$ be as in the proof of~\Cref{lem:existence} --- thus $P, Q \in [n]^{(h)}, X \in [n]^{(n - 2h)}, f(\2_P) = f(\2_Q)$ and $P, Q, X$ form a partition of $[n]$. Recall that $|X| \leq \ell$.

  Now, since the sets in $\mathcal{F}$ are disjoint, at most $\ell$ of the sets in
  $\mathcal{F}$ intersect $X$. Thus, since $|\mathcal{F}| \geq \ell+1$, there is
  some $Z\in \mathcal{F}$ disjoint from $X$. Let $C$ be any other set in
  $\mathcal{F}$ and define $C' = C \cup X$. Note that $|C'| \leq (t+1)\ell$.
  Note also that $f( \1_Z) \notin T$ as $Z$ is $T$-avoiding, hence $|T\cup \{f(\1_Z)\}| = t+1$. We will show that $C'$
  is $(T\cup \{f(\1_Z)\})$-avoiding, proving the result.

Since $C$ is $T$-avoiding and $C\subseteq C'$, $C'$ is also $T$-avoiding
  (cf.~\Cref{item:upwardsClosed}). To see that $C'$ is
  $\{f(\1_Z)\}$-avoiding, note that 
  for every $D' \subseteq [n]$ with $C'
  \subseteq D'$, 
   all the columns of the $r\times n$ matrix whose rows are $\1_{D'}$, $\2_P, \2_Q$, and $r-3$ many copies of $\1_Z$
    are tuples of $\LO_3^r$.
For $i \in P$ (respectively $i \in Q$), the
  unique maximum is given by a 2 in the $\2_P$ row (respectively the $\2_Q$ row). Since $P, Q, X$
  form a partition of $[n]$, it remains to consider $i \in X$. Since $P, Q, Z$ are
  disjoint from $X$, all the $\2_P, \2_Q, \1_Z$ rows have a 0 in these columns.
  On the other hand, the (unique) $\1_{D'}$ row has a 1, since $i \in X
  \subseteq C' \subseteq D'$. Hence we see that every column in the matrix is a
  tuple in $\LO_3^r$. Since $f$ is a polymorphism of
  $(\LO_3^r,\CF_\ell^r)$, $(f(\2_P), f(\2_Q),
  f(\1_{D'}), f(\1_Z), \ldots, f(\1_Z))$ is a tuple in $\CF_\ell^r$.

  Hence since $f(\2_P) = f(\2_Q)$, we have that $f(\1_{D'}) \neq f(\1_Z)$. Thus $C'$ is
  a $(T\cup \{f(\1_Z)\})$-avoiding, so $(t + 1)$-avoiding set as required.
\end{proof}

\begin{theorem}\label{locf}
$\PCSP(\LO_3^r, \CF_\ell^r)$ is \NP-hard for all constant $\ell \geq 3$ and $r\geq 4$.
\end{theorem}
\begin{proof}
  By~\Cref{our hardness condition}, \Cref{lem:existence} and~\Cref{lem:boosting}.
\end{proof}
\Cref{thm:cf} and~\Cref{thm:lo} follow immediately:
\cf*
\begin{proof}
The case $r=3$ is given by~\Cref{thm:main} as $\NAE^3_k=\CF^3_k$ for every $k\geq 2$.
The case $r\geq 5$ and $2\leq k\leq\ell$ follows from~\Cref{thm:main} and~\Cref{lem:red}.
For $r=4$ and $k=2$, recall from~\Cref{sec:intro} that $\CF_2^4$ is identical to
  solving systems of mod-2 equations of the form $x + y + z + t \equiv 1 \pmod
  2$. Thus, $\CSP(\CF_2^4)$ and also $\PCSP(\CF_2^4, \CF_\ell^4)$ is in \P{} by homomorphic relaxation.
For $r=4$ and $3\leq k\leq\ell$, the result follows by~\Cref{locf} and by homomorphic relaxation.
\end{proof}
    
\lo*

\begin{proof}
By~\Cref{locf} and by homomorphic relaxation as $\LO_3^r\to\LO_k^r$ and $\LO_\ell^r\to\CF_\ell^r$.
\end{proof}

\subsection{Extending hardness to other templates}
In the above proofs of~\Cref{lem:existence} and~\Cref{lem:boosting}, the only required property of $\CF_\ell^r$ is that for every tuple in the relation, if the first two entries are equal, then the remaining $r-2$ entries in the tuple cannot all be equal. Thus, the same proof also shows a stronger result:

\begin{definition}
Define $\BNAE_\ell^{s,r-s}$ 
    (Block-NAE) as the relational structure with domain $\{0,1,\dots,\ell-1\}$, with a single $r$-ary relation which contains the tuples for which if any $s$ entries are the same, then the remaining $r-s$ entries cannot all be the same. (This relation with $s=2$ is strictly larger than the relation corresponding to $\CF_\ell^r$ when $r \geq 6$).
\end{definition}
     Then the exact same proof from \Cref{sec:proof} (replacing every occurrence of $\CF_\ell^r$ with $\BNAE_\ell^{2,r-2}$) shows that $\PCSP(\LO_3^r, \BNAE_\ell^{2,r-2})$ is \NP-hard for all constant $\ell \geq 3, r\geq 4$. 

    In fact, using Kneser hypergraphs and the same proof technique, one can also show \NP-hardness of $\PCSP(\LO_3^r, \BNAE_\ell^{s,r-s})$ for all constant $\ell \geq 3, r\geq 4, 2\leq s \leq r-2$.

\begin{definition}
The \emph{$s$-uniform Kneser hypergraph} is defined as 
$\KG^{(s)}(A,h) = (A^{(h)},E)$,
where $\{P_1,\dots,P_s\} \in E$ if and only if $\{P_1,\dots,P_s\}$ is a collection of pairwise disjoint subsets of $A$ of size $h$. Again, for the special case $A = [n]$, we use the notation $\KG^{(s)}(n, h) = \KG^{(s)}([n], h)$.
\end{definition}
 There is an analogue of Lov\'asz's theorem for Kneser hypergraphs~\cite{AFL86}.\footnote{For an $s$-uniform hypergraph $H$, the \emph{chromatic number} ($\chi(H)$) of $H$, is the smallest $k$ such
that $H\to \NAE^s_k$.}

 \begin{theorem}[\cite{AFL86}]\label{kneserchrombound2}
$\chi(\KG^{(s)}(n,h))=\left\lceil\frac{n-s(h-1)}{s-1}\right\rceil$ for any $n \in \mathbb{N}$,$ 1\leq h \leq n/s$.
 \end{theorem}
    
The proof of hardness then continues completely analogously (and is actually a generalisation of the proof in \Cref{sec:proof}), but is included for completeness.

\begin{lemma}\label{lem:existence2}
    Let $ \ell \geq 3$, $r\geq 4$, and $ 2 \leq s \leq r-2$.
    Then any $f \in \Pol^{(n)}(\LO_3^r, \BNAE_\ell^{s,r-s})$ has a 1-avoiding set of size $\leq \ell(s-1)$.
\end{lemma}
\begin{proof}
  By~\Cref{item:domainAvoiding}, $[n]$ is an $(\ell-1)$-avoiding set and hence a 1-avoiding set. Thus if $n\leq\ell(s-1)$, we are done.
  Otherwise assume $n > \ell(s-1)$, and set $h = \lceil \frac{n-\ell(s-1)}{s}\rceil$.
    Consider the Kneser graph $\KG^{(s)}([n], h)$, and colour each vertex $P$ by $f(\2_P)$. By~\Cref{kneserchrombound2}, $\chi(\KG^{(s)}(n, h)) =\left\lceil\frac{n-s(h-1)}{s-1}\right\rceil > \ell$, so there exist disjoint sets $P_1,\dots,P_s \in [n]^{(h)}$  such that $f(\2_{P_1}) = f(\2_{P_2}) = \dots = f(\2_{P_s})$. Let $X \coloneqq [n]\setminus \bigcup_{j\in [s]}P_j$.
    
    Now observe that for any $Y \subseteq [n]$ with $X\subseteq Y$, 
     all the columns of the $r\times n$ matrix whose rows are $\1_Y$, $\2_{P_1}, \2_{P_2}, \dots, \2_{P_s}$, and $r-s-1$ many copies of $\0$
    are tuples of $\LO_3^r$.
    For $i \in P_j$, the unique maximum is a 2 in the $\2_{P_j}$ row; and for $i \in X$
    the unique maximum is a 1 in the $\1_Y$ row. Hence, since $f$ is a
    polymorphism of $(\LO_3^r,\BNAE_\ell^{s, r - s})$, we have that $(f(\1_Y), f(\2_{P_1}) , f(\2_{P_2}) , \dots , f(\2_{P_s}), f(\0), \ldots, f(\0))$ is a tuple in $\BNAE_\ell^{s,r-s}$. Since $f(\2_{P_1}) = f(\2_{P_2}) = \dots = f(\2_{P_s})$ we deduce that $f(\1_Y) \neq f(\0)$.
    Thus it follows that $X$ is $\{ f(\0) \}$-avoiding, and thus 1-avoiding. Noting that $X$ has size $n-sh\leq \ell(s-1)$ completes the proof.
\end{proof}

\begin{lemma} \label{lem:boosting2}

Let $\ell \geq 3, r\geq 4, 2 \leq s \leq r-2 $. Let $1\leq t<\ell$, and suppose that $f\in \Pol^{(n)}(\LO_3^r, \BNAE_\ell^{s,r-s})$ has a disjoint family of $> \binom{\ell}{t}(s-1)\ell$ many $t$-avoiding sets, all of size $\leq t(s-1)\ell$. Then $f$ has a $(t+1)$-avoiding set of size $\leq (t+1)(s-1)\ell$
  
\end{lemma}
\begin{proof}
By assumption and the pigeonhole principle, there is a family $\mathcal{F}$ of at least $(s-1)\ell+1$ disjoint $T$-avoiding sets of size $\leq t(s-1)\ell$ for some $T \subseteq \{0,1,\dots,\ell-1\}$ of size $t$. Let $h,P_1,\dots,P_s,X$ be as in the proof of~\Cref{lem:existence2}. Recall $|X| \leq (s - 1) \ell$.

  Now, since the sets in $\mathcal{F}$ are disjoint, at most $(s-1)\ell$ of the sets in
  $\mathcal{F}$ intersect $X$. Thus, since $|\mathcal{F}| \geq (s-1)\ell+1$, there is
  some $Z\in \mathcal{F}$ disjoint from $X$. Let $C$ be any other set in
  $\mathcal{F}$ and define $C' = C \cup X$. Note that $|C'| \leq (t+1)(s-1)\ell$.
  Note also that $f( \1_Z) \notin T$ as $Z$ is $T$-avoiding, hence $|T\cup \{f(\1_Z)\}| = t+1$. We will show that $C'$
  is $(T\cup \{f(\1_Z)\})$-avoiding, proving the result.

Since $C$ is $T$-avoiding and $C\subseteq C'$, $C'$ is also $T$-avoiding
  (cf.~\Cref{item:upwardsClosed}). To see that $C'$ is
  $\{f(\1_Z)\}$-avoiding, note that 
  for every $D' \subseteq [n]$ with $C'
  \subseteq D'$, 
   all the columns of the $r\times n$ matrix whose rows are $\1_{D'}$, $\2_{P_1}, \2_{P_2}, \dots, \2_{P_s}$, and $r-s-1$ many copies of $\1_Z$
    are tuples of $\LO_3^r$.
For $i \in P_j$, the
  unique maximum is given by a 2 in the $\2_{P_j}$ row. Since $P_1, \ldots, P_s, X$
  form a partition of $[n]$, it remains to consider $i \in X$. Since $P_1,\dots,P_s, Z$ are
  disjoint from $X$, the $\1_Z$ row and all the $\2_{P_j} $ rows have a 0 in these columns.
  On the other hand, the (unique) $\1_{D'}$ row has a 1, since $i \in X
  \subseteq C' \subseteq D'$. Hence we see that every column in the matrix is a
  tuple in $\LO_3^r$. Since $f$ is a polymorphism of
  $(\LO_3^r,\BNAE_\ell^{s,r-s})$, $(f(\2_{P_1}) , f(\2_{P_2}) , \dots , f(\2_{P_s}),
  f(\1_{D'}), f(\1_Z), \ldots, f(\1_Z))$ is a tuple in $\BNAE_\ell^{s,r-s}$.

  Hence since $f(\2_{P_1}) = f(\2_{P_2}) = \dots = f(\2_{P_s})$, we have that $f(\1_{D'}) \neq f(\1_Z)$. Thus $C'$ is
  a $(T\cup \{f(\1_Z)\})$-avoiding, so $(t + 1)$-avoiding set as required.
\end{proof}

\begin{theorem}\label{lobnae}
$\PCSP(\LO_3^r, \BNAE_\ell^{s,r-s})$ is \NP-hard for all constant $\ell \geq 3$, $r\geq 4$, and $2\leq s \leq r-2$.
\end{theorem}
\begin{proof}
  By~\Cref{our hardness condition}, \Cref{lem:existence2} and~\Cref{lem:boosting2}.
\end{proof}

We finish with noting that our proofs can be used to show \NP-hardness of more
templates if the symmetry requirement is dropped, then capturing colouring
variants of directed hypergraphs.

\section*{Acknowledgement}

We thank the anonymous reviewers for their detailed feedback.

{\small
\bibliographystyle{alphaurl}
\bibliography{bib}

\newcommand{\etalchar}[1]{$^{#1}$}
\begin{thebibliography}{KOW{\v{Z}}23}

\bibitem[ABP20]{ABP20}
Per Austrin, Amey Bhangale, and Aditya Potukuchi.
\newblock Improved inapproximability of rainbow coloring.
\newblock In {\em Proc. 31st Annual {ACM-SIAM} Symposium on Discrete Algorithms
  (SODA'20)}, pages 1479--1495, 2020.
\newblock \href {https://arxiv.org/abs/1810.02784} {\path{arXiv:1810.02784}},
  \href {https://doi.org/10.1137/1.9781611975994.90}
  {\path{doi:10.1137/1.9781611975994.90}}.

\bibitem[AFL86]{AFL86}
N.~Alon, P.~Frankl, and L.~Lov{\'a}sz.
\newblock The {{Chromatic Number}} of {{Kneser Hypergraphs}}.
\newblock {\em Trans. Am. Math. Soc.}, 298(1):359--370, 1986.
\newblock \href {https://arxiv.org/abs/2000624} {\path{arXiv:2000624}}, \href
  {https://doi.org/10.2307/2000624} {\path{doi:10.2307/2000624}}.

\bibitem[AGH17]{AGH17}
Per Austrin, Venkatesan Guruswami, and Johan H{å}stad.
\newblock (2+{$\epsilon$})-{S}at is {NP}-hard.
\newblock {\em {SIAM} J. Comput.}, 46(5):1554--1573, 2017.
\newblock \href {https://doi.org/10.1137/15M1006507}
  {\path{doi:10.1137/15M1006507}}.

\bibitem[BBB21]{Barto21:stacs}
Libor Barto, Diego Battistelli, and Kevin~M. Berg.
\newblock {Symmetric Promise Constraint Satisfaction Problems: Beyond the
  Boolean Case}.
\newblock In {\em Proc. 38th International Symposium on Theoretical Aspects of
  Computer Science (STACS'21)}, volume 187 of {\em LIPIcs}, pages 10:1--10:16,
  2021.
\newblock \href {https://arxiv.org/abs/2010.04623} {\path{arXiv:2010.04623}},
  \href {https://doi.org/10.4230/LIPIcs.STACS.2021.10}
  {\path{doi:10.4230/LIPIcs.STACS.2021.10}}.

\bibitem[BBKO21]{BBKO21}
Libor Barto, Jakub Bulín, Andrei~A. Krokhin, and Jakub Opršal.
\newblock Algebraic approach to promise constraint satisfaction.
\newblock {\em J. {ACM}}, 68(4):28:1--28:66, 2021.
\newblock \href {https://arxiv.org/abs/1811.00970} {\path{arXiv:1811.00970}},
  \href {https://doi.org/10.1145/3457606} {\path{doi:10.1145/3457606}}.

\bibitem[BG16]{BG16-graph}
Joshua Brakensiek and Venkatesan Guruswami.
\newblock New hardness results for graph and hypergraph colorings.
\newblock In {\em Proc. 31st Conference on Computational Complexity (CCC'16)},
  volume~50 of {\em LIPIcs}, pages 14:1--14:27, 2016.
\newblock \href {https://doi.org/10.4230/LIPIcs.CCC.2016.14}
  {\path{doi:10.4230/LIPIcs.CCC.2016.14}}.

\bibitem[BG21]{BG21:sicomp}
Joshua Brakensiek and Venkatesan Guruswami.
\newblock {Promise Constraint Satisfaction: Algebraic Structure and a Symmetric
  Boolean Dichotomy}.
\newblock {\em {SIAM} J. Comput.}, 50(6):1663--1700, 2021.
\newblock \href {https://arxiv.org/abs/1704.01937} {\path{arXiv:1704.01937}},
  \href {https://doi.org/10.1137/19M128212X} {\path{doi:10.1137/19M128212X}}.

\bibitem[BW{\v{Z}}21]{BWZ21}
Alex Brandts, Marcin Wrochna, and Stanislav {\v{Z}}ivn{\'{y}}.
\newblock The complexity of promise {SAT} on non-{B}oolean domains.
\newblock {\em {ACM} Trans. Comput. Theory}, 13(4):26:1--26:20, 2021.
\newblock \href {https://arxiv.org/abs/1911.09065} {\path{arXiv:1911.09065}},
  \href {https://doi.org/10.1145/3470867} {\path{doi:10.1145/3470867}}.

\bibitem[CKP13]{Cheilaris13:sidma}
Panagiotis Cheilaris, Balázs Keszegh, and Dömötör Pálvölgyi.
\newblock Unique-maximum and conflict-free coloring for hypergraphs and tree
  graphs.
\newblock {\em {SIAM} J. Discret. Math.}, 27(4):1775--1787, 2013.
\newblock \href {https://doi.org/10.1137/120880471}
  {\path{doi:10.1137/120880471}}.

\bibitem[DMR09]{Dinur09:sicomp}
Irit Dinur, Elchanan Mossel, and Oded Regev.
\newblock {Conditional Hardness for Approximate Coloring}.
\newblock {\em {SIAM} J. Comput.}, 39(3):843--873, 2009.
\newblock \href {https://doi.org/10.1137/07068062X}
  {\path{doi:10.1137/07068062X}}.

\bibitem[DRS05]{DRS05}
Irit Dinur, Oded Regev, and Clifford Smyth.
\newblock {The Hardness of 3-Uniform Hypergraph Coloring}.
\newblock {\em Comb.}, 25(5):519--535, 2005.
\newblock \href {https://doi.org/10.1007/s00493-005-0032-4}
  {\path{doi:10.1007/s00493-005-0032-4}}.

\bibitem[ELRS03]{Even03:sicomp}
Guy Even, Zvi Lotker, Dana Ron, and Shakhar Smorodinsky.
\newblock Conflict-free colorings of simple geometric regions with applications
  to frequency assignment in cellular networks.
\newblock {\em {SIAM} J. Comput.}, 33(1):94--136, 2003.
\newblock \href {https://doi.org/10.1137/S0097539702431840}
  {\path{doi:10.1137/S0097539702431840}}.

\bibitem[FNO{\etalchar{+}}25]{fnotw25:toct}
Marek Filakovsk\'{y}, Tamio-Vesa Nakajima, Jakub Opr\v{s}al, Gianluca Tasinato,
  and Uli Wagner.
\newblock {Hardness of Linearly Ordered 4-Colouring of 3-Colourable 3-Uniform
  Hypergraphs}.
\newblock {\em {ACM} Trans. Comput. Theory}, 2025.
\newblock \href {https://arxiv.org/abs/2312.12981} {\path{arXiv:2312.12981}},
  \href {https://doi.org/10.1145/3779121} {\path{doi:10.1145/3779121}}.

\bibitem[GHS02]{GHS02-HardnessApproximateHypergraphColoring}
Venkatesan Guruswami, Johan H{\aa}stad, and Madhu Sudan.
\newblock Hardness of {{Approximate Hypergraph Coloring}}.
\newblock {\em {SIAM} J. Comput.}, 31(6):1663--1686, 2002.
\newblock \href {https://doi.org/10.1137/S0097539700377165}
  {\path{doi:10.1137/S0097539700377165}}.

\bibitem[GJ76]{GJ76}
M.~R. Garey and D.~S. Johnson.
\newblock {The Complexity of Near-Optimal Graph Coloring}.
\newblock {\em J. {ACM}}, 23(1):43--49, 1976.
\newblock \href {https://doi.org/10.1145/321921.321926}
  {\path{doi:10.1145/321921.321926}}.

\bibitem[GL18]{GL18}
Venkatesan Guruswami and Euiwoong Lee.
\newblock Strong inapproximability results on balanced rainbow-colorable
  hypergraphs.
\newblock {\em Comb.}, 38(3):547--599, 2018.
\newblock \href {https://doi.org/10.1007/s00493-016-3383-0}
  {\path{doi:10.1007/s00493-016-3383-0}}.

\bibitem[GS20a]{GS20:icalp}
Venkatesan Guruswami and Sai Sandeep.
\newblock {d-To-1 Hardness of Coloring 3-Colorable Graphs with {O(1)} Colors}.
\newblock In {\em Proc. 47th International Colloquium on Automata, Languages,
  and Programming (ICALP'20)}, volume 168 of {\em LIPIcs}, pages 62:1--62:12.
  Schloss Dagstuhl -- Leibniz-Zentrum f{\"u}r Informatik, 2020.
\newblock \href {https://doi.org/10.4230/LIPIcs.ICALP.2020.62}
  {\path{doi:10.4230/LIPIcs.ICALP.2020.62}}.

\bibitem[GS20b]{GS20:rainbow}
Venkatesan Guruswami and Sai Sandeep.
\newblock Rainbow coloring hardness via low sensitivity polymorphisms.
\newblock {\em SIAM J. Discret. Math}, 34(1):520--537, 2020.
\newblock \href {https://doi.org/10.1137/19M127731X}
  {\path{doi:10.1137/19M127731X}}.

\bibitem[Kar72]{Karp1972}
Richard~M. Karp.
\newblock Reducibility among combinatorial problems.
\newblock In {\em Complexity of Computer Computations: Proceedings of a
  symposium on the Complexity of Computer Computations}, pages 85--103.
  Springer US, 1972.
\newblock \href {https://doi.org/10.1007/978-1-4684-2001-2_9}
  {\path{doi:10.1007/978-1-4684-2001-2_9}}.

\bibitem[Kho02]{Khot02stoc}
Subhash Khot.
\newblock On the power of unique 2-prover 1-round games.
\newblock In {\em Proc. 34th Annual ACM Symposium on Theory of Computing
  ({STOC}'02)}, pages 767--775. ACM, 2002.
\newblock \href {https://doi.org/10.1145/509907.510017}
  {\path{doi:10.1145/509907.510017}}.

\bibitem[KOW{\v{Z}}23]{KOWZ23}
Andrei~A. Krokhin, Jakub Opršal, Marcin Wrochna, and Stanislav
  {\v{Z}}ivn{\'{y}}.
\newblock Topology and adjunction in promise constraint satisfaction.
\newblock {\em {SIAM} J. Comput.}, 52(1):37--79, 2023.
\newblock \href {https://arxiv.org/abs/2003.11351} {\path{arXiv:2003.11351}},
  \href {https://doi.org/10.1137/20M1378223} {\path{doi:10.1137/20M1378223}}.

\bibitem[KV25]{krokhin_approximating_2025}
Andrei Krokhin and Danny Vagnozzi.
\newblock Approximating 1-in-3 {SAT} by linearly ordered hypergraph 3-colouring
  is {NP}-hard, 2025.
\newblock \href {https://arxiv.org/abs/2508.14606} {\path{arXiv:2508.14606}}.

\bibitem[Lov78]{Lovasz78}
L\'aszl\'o Lov{\'a}sz.
\newblock Kneser's conjecture, chromatic number, and homotopy.
\newblock {\em J. Comb. Theory, Series A}, 25(3):319--324, November 1978.
\newblock \href {https://doi.org/10.1016/0097-3165(78)90022-5}
  {\path{doi:10.1016/0097-3165(78)90022-5}}.

\bibitem[NVW{\v{Z}}25]{nvwz25:icalp}
Tamio{-}Vesa Nakajima, Zephyr Verwimp, Marcin Wrochna, and Stanislav
  {\v{Z}}ivn{\'{y}}.
\newblock Complexity of approximate conflict-free, linearly-ordered, and
  nonmonochromatic hypergraph colourings.
\newblock In {\em Proc. 52nd International Colloquium on Automata, Languages,
  and Programming (ICALP'25)}, volume 334 of {\em LIPIcs}, pages 169:1--169:10.
  Schloss Dagstuhl - Leibniz-Zentrum f{\"{u}}r Informatik, 2025.
\newblock \href {https://doi.org/10.4230/LIPICS.ICALP.2025.169}
  {\path{doi:10.4230/LIPICS.ICALP.2025.169}}.

\bibitem[N{\v{Z}}23]{NZ23:toct}
Tamio-Vesa Nakajima and Stanislav {\v{Z}}ivn{\'{y}}.
\newblock {Linearly Ordered Colourings of Hypergraphs}.
\newblock {\em {ACM} Trans. Comput. Theory}, 13(3--4), 2023.
\newblock \href {https://arxiv.org/abs/2204.05628} {\path{arXiv:2204.05628}},
  \href {https://doi.org/10.1145/3570909} {\path{doi:10.1145/3570909}}.

\bibitem[Shi19]{shitov_counterexamples_2019}
Yaroslav Shitov.
\newblock Counterexamples to {Hedetniemi}'s conjecture.
\newblock {\em Ann. Math.}, 190(2), September 2019.
\newblock \href {https://doi.org/10.4007/annals.2019.190.2.6}
  {\path{doi:10.4007/annals.2019.190.2.6}}.

\bibitem[Smo13]{Smorodinsky2013}
Shakhar Smorodinsky.
\newblock Conflict-free coloring and its applications.
\newblock In Imre B{\'a}r{\'a}ny, K{\'a}roly~J. B{\"o}r{\"o}czky,
  G{\'a}bor~Fejes T{\'o}th, and J{\'a}nos Pach, editors, {\em Geometry ---
  Intuitive, Discrete, and Convex: A Tribute to L{\'a}szl{\'o} Fejes T{\'o}th},
  pages 331--389, Berlin, Heidelberg, 2013. Springer Berlin Heidelberg.
\newblock \href {https://arxiv.org/abs/1005.3616} {\path{arXiv:1005.3616}},
  \href {https://doi.org/10.1007/978-3-642-41498-5_12}
  {\path{doi:10.1007/978-3-642-41498-5_12}}.

\end{thebibliography}
}
\end{document}